\documentclass[conference,a4paper]{IEEEtran}

\usepackage{soul}

\usepackage{xcolor}
\usepackage{overpic}
\usepackage{balance}
\usepackage{graphicx}

 \usepackage{amsmath,amssymb,amsthm}

\ifCLASSOPTIONcompsoc
 \usepackage[caption=false,font=normalsize,labelfont=sf,textfont=sf]{subfig}
\else
 \usepackage[caption=false,font=footnotesize]{subfig}
\fi

\def\Htran{\mbox{\tiny $\mathrm{H}$}}
\def\Ttran{\mbox{\tiny $\mathrm{T}$}}
\def\CN{\mathcal{N}_{\mathbb{C}}} 

 \newtheorem{lemma}{Lemma}
 \newtheorem{corollary}{Corollary}
 \newtheorem{theorem}{Theorem}
 \theoremstyle{remark}
 \newtheorem{remark}{Remark}

\IEEEoverridecommandlockouts

\begin{document}

\title{Beamfocusing Versus Airy Beams:\\ Does Antenna Blockage Require Beam Bending?}

\author{\IEEEauthorblockN{
Emil Björnson\IEEEauthorrefmark{1},   
Vitaly Petrov\IEEEauthorrefmark{1},  
Erik G. Larsson\IEEEauthorrefmark{2}   
}                                     
\IEEEauthorblockA{\IEEEauthorrefmark{1}
Department of Communication Systems, KTH Royal Institute of Technology, Stockholm, Sweden, \{emilbjo,vitalyp\}@kth.se}%
\IEEEauthorblockA{\IEEEauthorrefmark{2}
Dept. of Electrical Engineering (ISY), Linköping University, Linköping, Sweden, erik.g.larsson@liu.se}%
 \thanks{The research was supported in part by Grant 2022-04222 from the Swedish
Research Council, ELLIIT, and the KAW foundation. GPT 6 Astra assisted with the literature review, manuscript drafting, simulation coding, and proofreading, but the authors take full responsibility for the concepts and content.}
}

\maketitle

\begin{abstract}
Airy beams are often proposed in array beamforming to bend beams around obstacles. We prove that, for a single-antenna receiver, conventional beamfocusing maximizes the received power regardless of how many antennas are blocked and, in expectation, even with phase errors in the channel state information. This also holds for multi-antenna receivers under a phase-separability condition that is satisfied in rank-one channels. Numerical results confirm that unmodified beamfocusing is optimal for physically small receivers (i.e., those that fit within the beamwidth) when an obstacle is present. For physically larger receivers, the optimal precoding deviates from beamfocusing to better illuminate the receiver's unobstructed parts, but will not become Airy beams. While Airy beams can be configured to perform well, their ability to follow a bent trajectory in front of an obstacle does not maximize the received power behind it.

\end{abstract}
\vskip0.5\baselineskip

\vspace{-1mm}

\vspace{-1mm}
\section{Introduction}
\label{sec:intro}
\vspace{-1mm}
The radiative near-field of an antenna array is the region where it must consider the nonlinear wavefront curvature when transmitting to a receiver, even if the receiver is in the far-field of the individual antenna elements~\cite{Bjornson2021b,Liu2023a}. This regime is particularly relevant to millimeter-wave (mmWave) and terahertz (THz) systems, where short wavelengths and electrically large apertures with very many antennas may place most users in the near field~\cite{Liu2023a,Singh2024a}. In a single-user link with known channel coefficients, maximum-ratio transmission (MRT) maximizes the signal-to-noise ratio (SNR) subject to a total-power constraint~\cite{Uchimura2026a}. In line-of-sight (LOS) scenarios, its phase shifts compensate for the individual path distances, producing \emph{beamfocusing} with limited depth in the near field, rather than only steering in an angular direction~\cite{Bjornson2021b}. Analog precoding, typical to mmWave and THz systems, retains this phase compensation but uses equal power per antenna.

\emph{Wavefront engineering} refers to actively shaping the spatial evolution of the transmitted field, especially in the near-field, by controlling phase and amplitude over the transmitting array~\cite{Singh2024a}. Examples include Bessel beams, proposed for extended depth of focus (compared to beamfocusing) and self-healing after partial obstruction~\cite{Singh2024a,Bodet2024a}, and Airy beams, associated with curved mainlobe trajectories and self-healing~\cite{Efremidis2019a}. These features can only be partially realized with finite-sized apertures and hardware constraints (e.g., analog precoding), but a closed-form phase-only Airy-beam design was recently developed in~\cite{Zhao2026a}.

Blockage provides a prominent motivation for these self-healing beam designs and is a practical issue in mmWave/THz systems~\cite{Uchimura2026a}. We use \emph{quasi-LOS} to describe LOS links where an object obstructs the direct paths from some transmit antennas to the receiver while leaving others unaffected. Airy beams have been particularly proposed for enhancing connectivity in these scenarios, with claims that they ``\emph{realize a link by curving around an intervening obstacle}''~\cite{Guerboukha2024a} and ``\emph{curve around the LOS obstruction and deliver maximum power to the receiver}''~\cite{Chen2025a}. However, the usefulness of Airy beams in such scenarios was challenged in~\cite{Schwan2026a}, concluding that the ``\emph{beamfocusing strategy already attains near-optimal performance}''.

These developments motivate the research question addressed in this paper: \textbf{When is it desirable to switch from beamfocusing to other beams (e.g., Airy) in the presence of a partially blocking object?} We investigate this question for quasi-LOS links with analog phase-only precoding and an explicit analytically tractable antenna-attenuation model. We characterize the SNR-maximizing precoding with and without phase uncertainty at the transmitter, examine extensions to multiple receive antennas, and relate Airy-beam phase profiles to their curved spatial trajectories. The results provide a precise answer to the question under specific conditions. Simulations are used to compare beamfocusing and Airy beams to corroborate the analytical results, explore their generality, and provide further insights into the key properties of these beam types.

\vspace{-1mm}
\section{Basic System Model}
\label{sec:system_model}
\vspace{-1mm}

Consider a narrowband multiple-input single-output (MISO) link with $M$ transmit antennas and a single receive antenna. The transmitter is equipped with a uniform linear array (ULA) with inter-antenna spacing $\Delta$. Antenna $m$ is located at $(x_m,0)$, where $x_m=(m-(M+1)/2)\Delta$. The receiver is located at $(r\sin(\theta),r\cos(\theta))$ in the radiative near-field of the array, where $r$ is the distance from the array center and $\theta$ is measured from broadside. The path distance from antenna $m$ is
\vspace{-0.5mm}
\begin{equation}
 r_m=\sqrt{r^2+x_m^2-2rx_m\sin(\theta)}.
 \label{eq:distance}
\end{equation}
If there is an unobstructed LOS link, the channel can be represented by 
$\mathbf h = [h_1,\ldots,h_M]^{\Ttran} \in \mathbb{C}^M$,
defined as $h_m=\alpha_m e^{-\mathrm j k r_m}$, using the wavenumber $k=\frac{2\pi}{\lambda}$, wavelength $\lambda$, and channel amplitudes  $\alpha_m >0$.
Note that this is a general non-uniform spherical channel wave model that can describe either isotropic antenna scenarios with $\alpha_m = \frac{\lambda}{4\pi r_m}$ or more detailed models with additional antenna/channel gain variations.

The transmitter sends $\mathbf x=\sqrt{P}\mathbf w s$, where $s$ is the communication signal with $\mathbb E\{|s|^2\}=1$ and $P$ is the total transmit power. 
As we aim to analyze the fundamentals of mmWave and THz systems, we consider analog precoding with continuously adjustable phases and fixed, equal amplitudes. Hence, the precoder $\mathbf w = [w_1,\ldots,w_M]^{\Ttran}$ can be selected from the set
\vspace{-0.5mm}
\begin{equation}
 \mathcal W=\left\{\mathbf w\in\mathbb C^M:
 |w_m|=\frac{1}{\sqrt M},\ m=1,\ldots,M\right\}.
 \label{eq:feasible}
\end{equation}
Consequently, the squared $L_2$-norm is $\|\mathbf w\|_2^2=1$, and each antenna transmits with power $P/M$. The received signal is
\vspace{-0.5mm}
\begin{equation}
 y=\sqrt P\,\mathbf h^{\Ttran}\mathbf w s+n,
 \label{eq:received}
\end{equation}
where $n\sim\CN(0,\sigma_n^2)$ is the noise with variance $\sigma_n^2$.
Assuming coherent reception, the signal-to-noise ratio (SNR) is
\vspace{-0.5mm}
\begin{equation}
 \gamma(\mathbf w) =\frac{P}{\sigma_n^2} |\mathbf h^{\Ttran}\mathbf w |^2.
 \label{eq:snr_unobstruct}
\end{equation}
Among precoders $\mathbf w \in \mathcal W$, the SNR is maximized by
\vspace{-0.5mm}
\begin{equation}
 [\mathbf w_{\mathrm{foc}}]_m
 =\frac{e^{-\mathrm j\arg(h_m)}}{\sqrt M}
 =\frac{e^{\mathrm jkr_m}}{\sqrt M}.
 \label{eq:focusing}
\end{equation}
This precoder design compensates for individual channel phases, so that all antenna contributions add coherently at the receiver, and is often referred to as \emph{beamfocusing}.
This term refers to the fact that the 3 dB focus area around the receiver is confined in both angle and depth, in parts of the near-field~\cite{Bjornson2021b}.

\vspace{-1mm}
\section{Optimal Precoding under Antenna Blockage}
\label{sec:blockage}
\vspace{-1mm}

We want to analyze precoding in a quasi-LOS scenario where the direct paths from a subset of the transmit antennas to the receiver are obstructed by a physical object.
We represent this by $\mathbf D=\operatorname{diag}(d_1,\ldots,d_M)$, where $d_m\in[0,1]$ represents the amplitude attenuation for antenna $m$.
Here, $d_m=0$ represents complete blockage, $d_m=1$ represents an unobstructed antenna, and intermediate values represent that the signal passes through the object, resulting in attenuation but no appreciable phase change. Note that we are not modeling diffraction, as this effect is weak in the mmWave/THz bands.

The resulting channel vector is $\mathbf D \mathbf h$ instead of $\mathbf h$, so the SNR in \eqref{eq:snr_unobstruct} can be generalized to
\vspace{-0.5mm}
\begin{equation}
 \gamma(\mathbf w;\mathbf D) = \frac{P}{\sigma_n^2} |\mathbf h^{\Ttran} \mathbf D^{\Ttran} \mathbf w |^2
 =\frac{P}{\sigma_n^2}
 \left|\sum_{m=1}^M d_m h_m w_m\right|^2.
 \label{eq:snr}
\end{equation}
The next lemma determines the optimal precoder for this case.

\begin{lemma}
\label{lem:blockage}
For every diagonal blockage matrix $\mathbf D$ with entries in $[0,1]$, the SNR in \eqref{eq:snr} is maximized over $\mathcal W$ by the same precoder $\mathbf w_{\mathrm{foc}}$ in \eqref{eq:focusing} as in the absence of blockage. The maximum SNR is
\vspace{-0.5mm}
\begin{equation}
 \gamma_{\max}(\mathbf D)
 =\frac{P}{M\sigma_n^2}
 \left(\sum_{m=1}^M d_m\alpha_m\right)^2.
 \label{eq:max_snr}
\end{equation}
\end{lemma}
\begin{IEEEproof}
For any $\mathbf w\in\mathcal W$, the triangle inequality gives
\vspace{-0.5mm}
\begin{equation}
 \left|\sum_{m=1}^M d_m h_m w_m\right|
 \leq\sum_{m=1}^M d_m|h_m||w_m|
 =\frac{1}{\sqrt M}\sum_{m=1}^M d_m\alpha_m.
\end{equation}
The precoder in \eqref{eq:focusing} attains equality because $d_m h_m[\mathbf w_{\mathrm{foc}}]_m=d_m\alpha_m/\sqrt M$ is real and nonnegative for every $m$. Squaring the bound and multiplying by $P/\sigma_n^2$ gives the final result.
\end{IEEEproof}

The lemma proves that there is no need to change the precoder when a blocking object appears, as this cannot improve the SNR.
Hence, the transmitter does not need to know $\mathbf D$ to select the SNR-maximizing phases, but only the nominal phases in $\mathbf h$.
However, to continue supporting data transmission, the receiver must feed back an SNR measurement so the data rate can be adjusted accordingly.

The physical intuition for this result is that the object attenuates certain antennas but does not affect the coherence of the wave components that reach the receiver.
If some antennas are entirely blocked (i.e., $d_m = 0$), they can continue transmitting as before, but the signal components will be absorbed or reflected by the object.
In this special case, the SNR-maximizing precoder is not unique, as the blocked antennas could switch their phases to anything they like.\footnote{To be precise, the original precoder in \eqref{eq:focusing} is not a unique solution, since a common phase rotation of all entries will leave the SNR unchanged.}

\subsection{Precoding with Phase-Uncertainty}
\label{subsec:uncertainty}

Next, we consider whether channel uncertainty affects the result in Lemma~\ref{lem:blockage}.
Suppose the transmitter knows the nominal vector $\mathbf h$, while the actual channel is $\tilde{\mathbf{h}} = [ \tilde h_1, \ldots,  \tilde h_M]^{\Ttran}$ with
\vspace{-0.5mm}
\begin{equation}
 \tilde h_m=h_m e^{\mathrm j\epsilon_m},\qquad
 \epsilon_m\overset{\mathrm{i.i.d.}}{\sim} \mathcal{N} (0,\sigma_\phi^2).
 \label{eq:phase_errors}
\end{equation}
The precoder must be selected without knowing the phase-error realizations $\epsilon_1, \ldots, \epsilon_M$, so we will maximize the average SNR.
To this end, we define $z_m=d_mh_mw_m$ and $\rho=e^{-\sigma_\phi^2}$, and notice that the error distribution implies that
\vspace{-0.5mm}
\begin{equation}
 \mathbb E\{e^{\mathrm j(\epsilon_m-\epsilon_\ell)}\}
 =\begin{cases}1,&m=\ell,\\ \rho,&m\ne\ell.
 \end{cases}
\end{equation}
For any deterministic precoder $\mathbf w$, the average SNR  becomes
\vspace{-0.5mm}
\begin{align}
 \overline\gamma(\mathbf w;\mathbf D)
 &=\frac{P}{\sigma_n^2}
 \mathbb E\left\{\left|\sum_{m=1}^M z_m e^{\mathrm j\epsilon_m}\right|^2\right\}\notag\\
 &=\frac{P}{\sigma_n^2}\left(
 \rho|\mathbf h^{\Ttran}\mathbf D\mathbf w|^2
 +(1-\rho)\sum_{m=1}^M d_m^2\alpha_m^2|w_m|^2\right).
 \label{eq:average_general}
\end{align}
As $|w_m|^2 \!=\!\frac{1}{M}$ for $\mathbf w\in\mathcal W$ and using $\gamma(\mathbf w;\mathbf D)$ in \eqref{eq:snr}, we have
\vspace{-0.5mm}
\begin{equation}
 \overline\gamma(\mathbf w;\mathbf D)
 =\rho\gamma(\mathbf w;\mathbf D)
 +\frac{P(1-\rho)}{M\sigma_n^2}
 \sum_{m=1}^M d_m^2\alpha_m^2.
 \label{eq:average_analog}
\end{equation}
The first term is proportional to the SNR without phase errors, and the second term is independent of the precoder. Hence, Lemma~\ref{lem:blockage} can be generalized as follows.

\begin{corollary}
\label{cor:uncertainty}
The precoder $\mathbf w_{\mathrm{foc}}$ in \eqref{eq:focusing} maximizes the average SNR in \eqref{eq:average_analog}, among all $\mathbf w \in \mathcal W$, for every blockage matrix $\mathbf D$ and every finite $\sigma_\phi^2\geq0$. The maximum is
\vspace{-0.5mm}
\begin{equation}
 \overline\gamma_{\max}(\mathbf D)
 \hspace{-0.5mm}=\hspace{-0.5mm}\frac{P}{M\sigma_n^2} \! \left[
 \rho\left(\sum_{m=1}^M d_m\alpha_m\right)^{\!2}
 \hspace{-1mm}+\hspace{-0.5mm}(1-\rho)\sum_{m=1}^M d_m^2\alpha_m^2\right]\!.
 \label{eq:average_max}
\end{equation}
\end{corollary}

This result implies that even uncertainty about the receiver location (or imperfect transmitter calibration) is not a reason to deviate from the beamfocusing precoder.
Admittedly, maximizing the average SNR is not the same as maximizing the rate, but these objectives are closely related in systems with a large number of antennas, since averaging over errors then appears naturally through law-of-large-numbers arguments.

\begin{remark}
The results in this section hold under more general conditions than we considered, because the ULA geometry and the non-uniform spherical-wave channel model were not used in the proofs. Similarly, it is sufficient that the phase errors are independent and have $\mathbb E\{e^{\mathrm j\epsilon_m}\}\geq 0$. 
By contrast, the results are not guaranteed to hold in the absence of per-antenna power constraints (since blocked antennas can then reallocate power to other antennas), discrete phase constraints (since blockage of an antenna affects the desired phase quantization at other antennas), or blockage-induced phase changes (if $d_m$ has a phase, antenna $m$ should change $w_m$ to compensate for it).
\end{remark}

\vspace{-2mm}
\section{Extension to Multiple Receive Antennas}
\label{sec:mimo}
\vspace{-1mm}

We next consider a multiple-input multiple-output (MIMO) channel with $N$ receive antennas, while retaining analog precoding with $\mathbf w \in\mathcal W$.
We let $\mathbf H=[\mathbf h_1,\ldots,\mathbf h_M]\in\mathbb C^{N\times M}$ denote the unobstructed channel, where $\mathbf h_m \in\mathbb C^{N}$ is the channel from transmit antenna $m$ to all receive antennas.
For simplicity, if the transmit antenna $m$ is obstructed, then this status affects all receive antennas.\footnote{This modeling is most applicable to the downlink direction, where the number of transmit antennas typically is much larger than the corresponding number of antenna elements at the receiving user device (i.e., $M \gg N$).}
We can then express the received signal $ \mathbf y \in \mathbb{C}^{N}$ as
\vspace{-0.5mm}
\begin{equation}
 \mathbf y=\sqrt P\,\mathbf H\mathbf D\mathbf w s+\mathbf n,
 \qquad \mathbf n\sim\CN(\mathbf 0,\sigma_n^2\mathbf I_N),
 \label{eq:mimo_received}
\end{equation}
where the blockage matrix $\mathbf D$ is defined as in Section~\ref{sec:blockage}.

\subsection{Fixed Receive Combining}

Suppose the receiver forms $\mathbf v^{\Htran}\mathbf y$ using a fixed combining vector $\mathbf v\in\mathbb C^N$ that satisfies $\| \mathbf v \|_2^2=1$.
The resulting SNR is
\vspace{-0.5mm}
\begin{equation}
 \gamma(\mathbf w,\mathbf v;\mathbf D)
 =\frac{P} {\sigma_n^2}  |\mathbf v^{\Htran}\mathbf H\mathbf D\mathbf w|^2.
 \label{eq:mimo_fixed_snr}
\end{equation}
By defining the scalar coefficients $g_m=\mathbf v^{\Htran}\mathbf h_m$, \eqref{eq:mimo_fixed_snr} becomes $(P/\sigma_n^2) |\sum_m d_m g_m w_m|^2$, which has the same structure as the MISO expression in \eqref{eq:snr}. We then obtain the following result.

\begin{corollary}
\label{cor:mimo_fixed}
For any fixed unit-norm receive combiner $\mathbf v$, the SNR in \eqref{eq:mimo_fixed_snr} is maximized over $\mathcal W$ by beamfocusing based on the effective channel $\mathbf v^{\Htran}\mathbf H = [g_1,\ldots,g_M]^{\Ttran}$ with\footnote{If $g_m=0$, the phase can be selected arbitrarily, e.g., $\arg(0)=0$.}
\vspace{-0.5mm}
\begin{equation}
 [\mathbf w_{\mathbf v}]_m
 =\frac{e^{-\mathrm j\arg(g_m)}}{\sqrt M},
 \label{eq:mimo_fixed_precoder}
\end{equation}
regardless of the value of $\mathbf D$.
The maximum SNR is
\vspace{-0.5mm}
\begin{equation}
 \gamma_{\max}(\mathbf v;\mathbf D)
 =\frac{P}{M\sigma_n^2}
 \left(\sum_{m=1}^M d_m|g_m|\right)^2.
\end{equation}
\end{corollary}

This result shows that the same beamfocusing precoder is optimal both before and during blockage, as long as the receive combiner remains the same in the blocked (quasi-LOS) and unblocked (LOS) scenarios. 
To perform beamfocusing, the transmitter must know the phases of $g_m$, but not~$\mathbf D$.

\subsection{Optimized Analog Receive Combining}
\label{subsec:analog_combining}

The receiver can benefit from revising its combiner after the blockage has occurred. 
We now optimize the receive combiner jointly with the transmit precoder, assuming an analog implementation with continuously adjustable phases (similar to the transmitter):
\vspace{-0.5mm}
\begin{equation}
 \mathcal V=\{\mathbf v\in\mathbb C^N:|v_n|=1/\sqrt{N},\ n=1,\ldots,N\}.
 \label{eq:mimo_receive_constraint}
\end{equation}
Since $\|\mathbf v\|_2^2=1$, the SNR in \eqref{eq:mimo_fixed_snr} becomes
\vspace{-0.5mm}
\begin{equation}
 \gamma(\mathbf w,\mathbf v;\mathbf D)
 =\frac{P}{\sigma_n^2}
 |\mathbf v^{\Htran}\mathbf H\mathbf D\mathbf w|^2.
 \label{eq:mimo_analog_snr}
\end{equation}
For a given $\mathbf w$, the optimal receive combiner aligns its phases element-by-element with $\mathbf q=\mathbf H\mathbf D\mathbf w$, yielding
\vspace{-0.5mm}
\begin{equation}
 \max_{\mathbf v\in\mathcal V}\gamma(\mathbf w,\mathbf v;\mathbf D)
 =\frac{P}{N\sigma_n^2}\|\mathbf H\mathbf D\mathbf w\|_1^2,
 \label{eq:mimo_analog_optimal_receive}
\end{equation}
using the $L_1$-norm. We will now derive a sufficient condition on the channel matrix for which one transmit--receive phase configuration is optimal for every blockage matrix.

\begin{theorem}
\label{thm:mimo_separable_phases}
Denote the $nm$:th entry of the channel matrix $\mathbf H$ as $H_{nm}$. Suppose all its entries have separable phases:
\vspace{-0.5mm}
\begin{equation}
 H_{nm}=A_{nm}e^{\mathrm j(\eta_n+\xi_m)},
 \qquad A_{nm}\geq0,
 \label{eq:mimo_separable_phases}
\end{equation}
where $\eta_n$ and $\xi_m$ are real. For every diagonal blockage matrix $\mathbf D$ with entries in $[0,1]$, the SNR in \eqref{eq:mimo_analog_snr} is jointly maximized over $\mathcal W$ and $\mathcal V$ by the same precoder and combiner pair:
\vspace{-0.5mm}
\begin{equation}
 w_m=\frac{e^{-\mathrm j\xi_m}}{\sqrt M},
 \qquad v_n=\frac{e^{\mathrm j\eta_n}}{\sqrt N}.
 \label{eq:mimo_separable_optimum}
\end{equation}
The maximum SNR is
$\frac{P}{NM\sigma_n^2}
 \left(\sum_{n=1}^N\sum_{m=1}^M d_m A_{nm}\right)^2$.
\end{theorem}
\begin{IEEEproof}
For any feasible $\mathbf w$ and $\mathbf v$, we obtain
\vspace{-0.5mm}
\begin{align}
 &|\mathbf v^{\Htran}\mathbf H\mathbf D\mathbf w|
 =\left|\sum_{n=1}^N\sum_{m=1}^M v_n^*H_{nm}d_mw_m\right|\notag\\
 &\leq\sum_{n=1}^N\sum_{m=1}^M |v_n||H_{nm}|d_m|w_m|=\frac{1}{\sqrt{NM}}\sum_{n=1}^N\sum_{m=1}^M d_mA_{nm}
\end{align}
using the triangle inequality.
The configuration in \eqref{eq:mimo_separable_optimum} makes each summand $v_n^*H_{nm}d_mw_m=d_mA_{nm}/\sqrt {NM}$ real and nonnegative, so it attains the bound for any value of $\mathbf D$. 
\end{IEEEproof}

The phase separability condition in \eqref{eq:mimo_separable_phases} might seem peculiar, but it is approximately satisfied in many scenarios.
A key example with exact separability is a rank-one channel matrix $\mathbf H=\mathbf a\mathbf b^{\Ttran}$, for which
$A_{nm}=|a_n||b_m|$, $\eta_n=\arg(a_n)$, and $\xi_m=\arg(b_m)$. 
Such channels appear exactly in plane-wave LOS channels (far-field in free space) and approximately in near-field scenarios where the receiver aperture fits within the transmitter's beamwidth; see Section~\ref{subsec:beamsMIMO} for an example.

\begin{figure*}[t!]
\centering
\vspace{-1mm}
\subfloat[\vspace{-1mm}Original Airy beam]{  \hspace{-2mm}
\begin{overpic}[trim={0 0 1.47cm 0.7cm},clip,width=0.5\columnwidth]{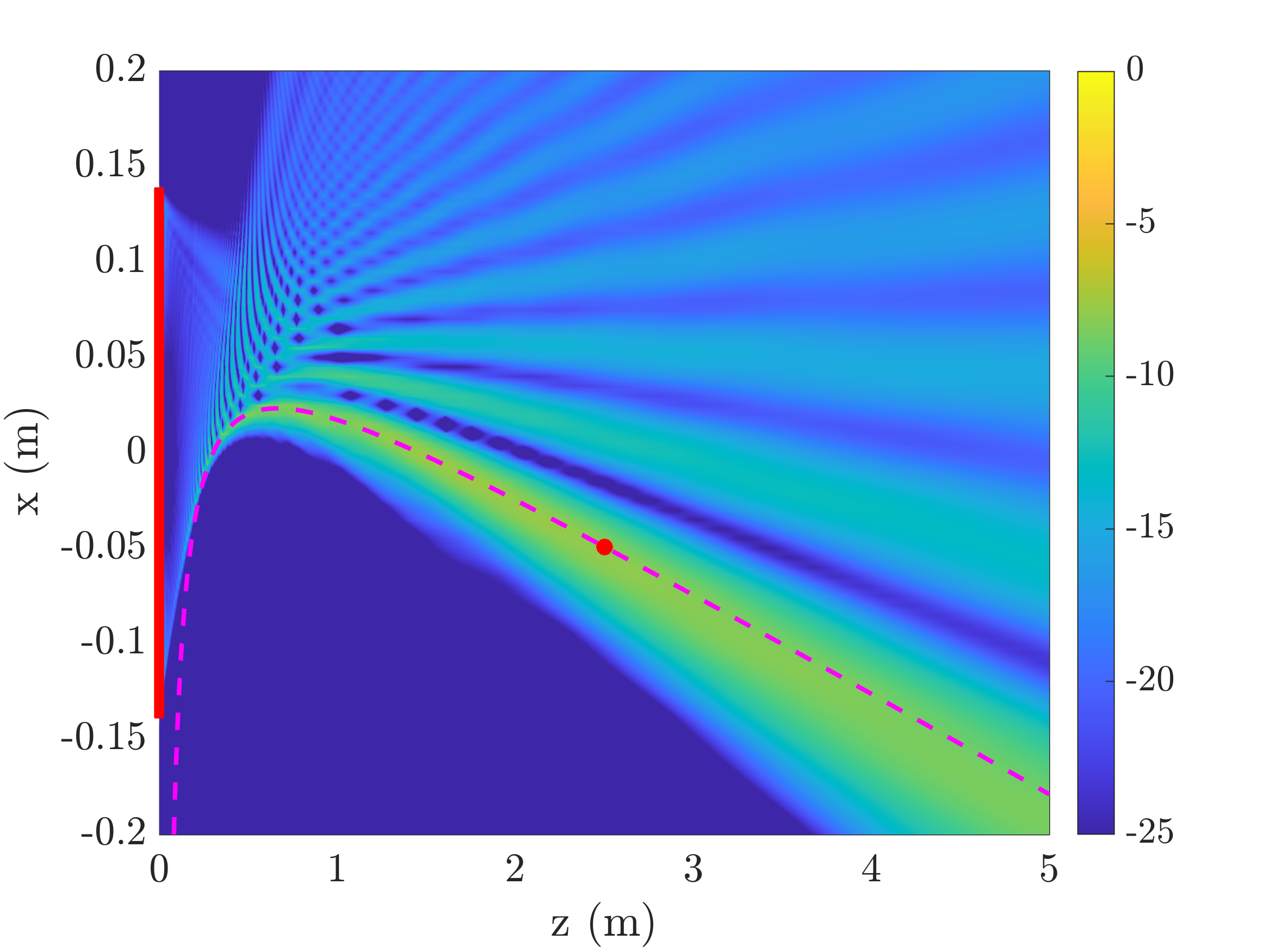}%
\put(61.5,44){\small \color{white} $-8.1$\,dB}
 \put(61,44){\color{white} \vector(-1,-1){8}}
\end{overpic}
\label{fig:blockage1}  \hspace{-1.2mm}}
\subfloat[\vspace{-1mm}Airy beam under blockage]{
\begin{overpic}[trim={0 0 1.47cm 0.7cm},clip,width=0.5\columnwidth]{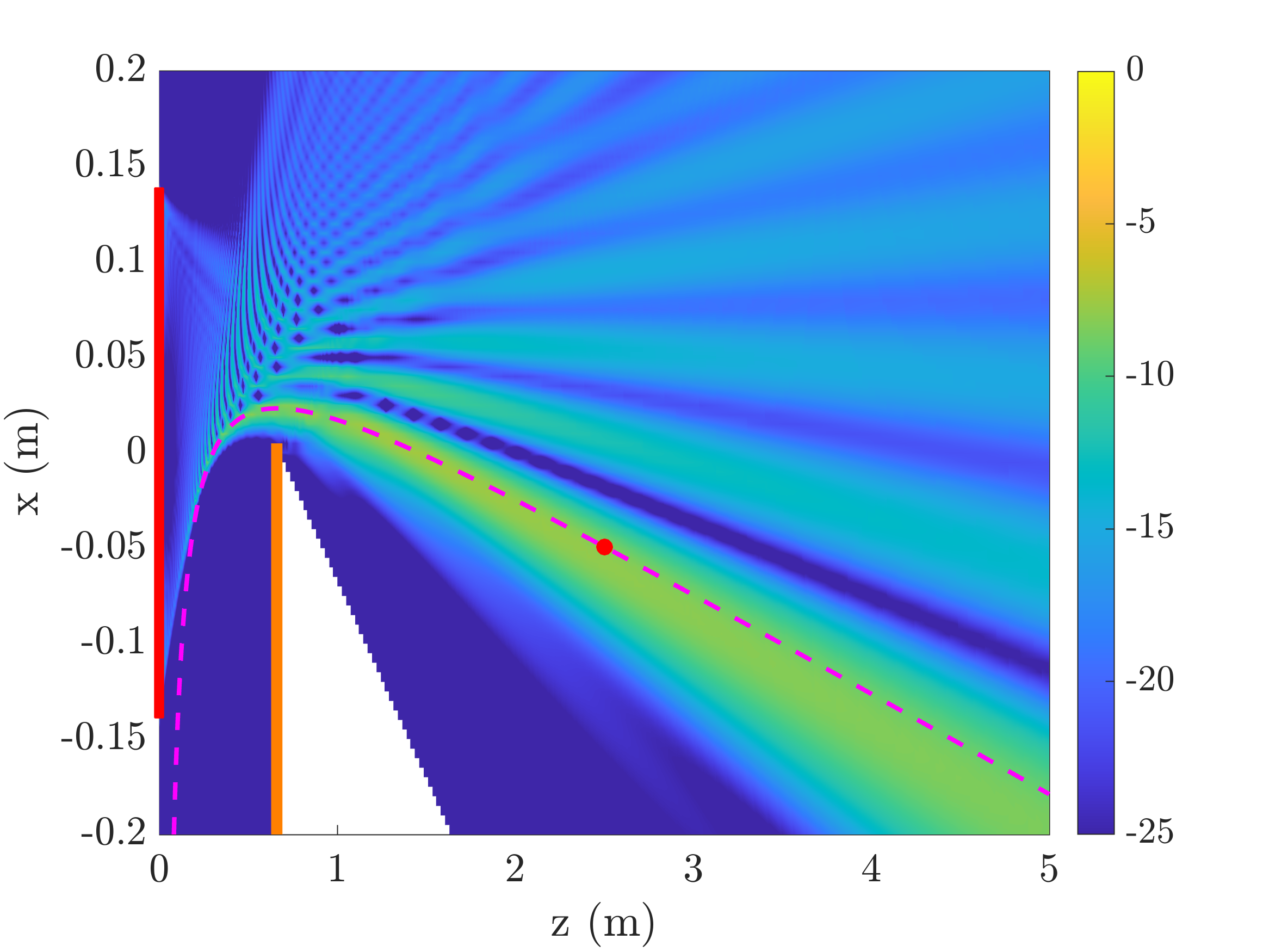}
\put(61.5,44){\small \color{white} $-8.2$\,dB}
 \put(61,44){\color{white} \vector(-1,-1){8}}
\end{overpic}
\label{fig:blockage2}  \hspace{-1.2mm}}
\subfloat[\vspace{-1mm}Original beamfocusing]{
\begin{overpic}[trim={0 0 1.47cm 0.7cm},clip,width=0.5\columnwidth]{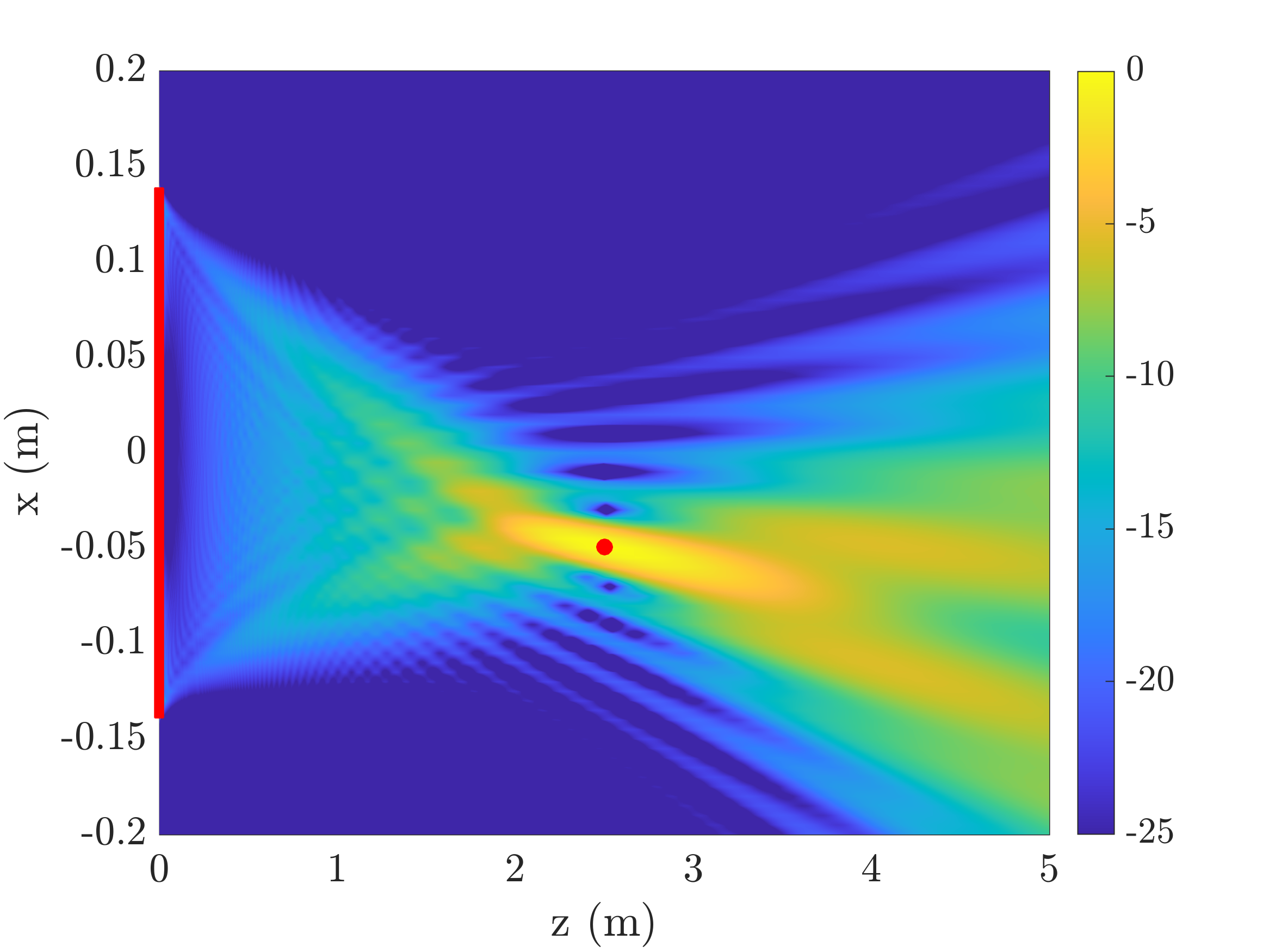}%
\put(61.5,44){\small \color{white} $-0.0$\,dB}
 \put(61,44){\color{white} \vector(-1,-1){8}}
\end{overpic}
\label{fig:blockage3} \hspace{-1.2mm}} 
\subfloat[\vspace{-1mm}Beamfocusing under blockage]{
\begin{overpic}[trim={0 0 1.47cm 0.7cm},clip,width=0.5\columnwidth]{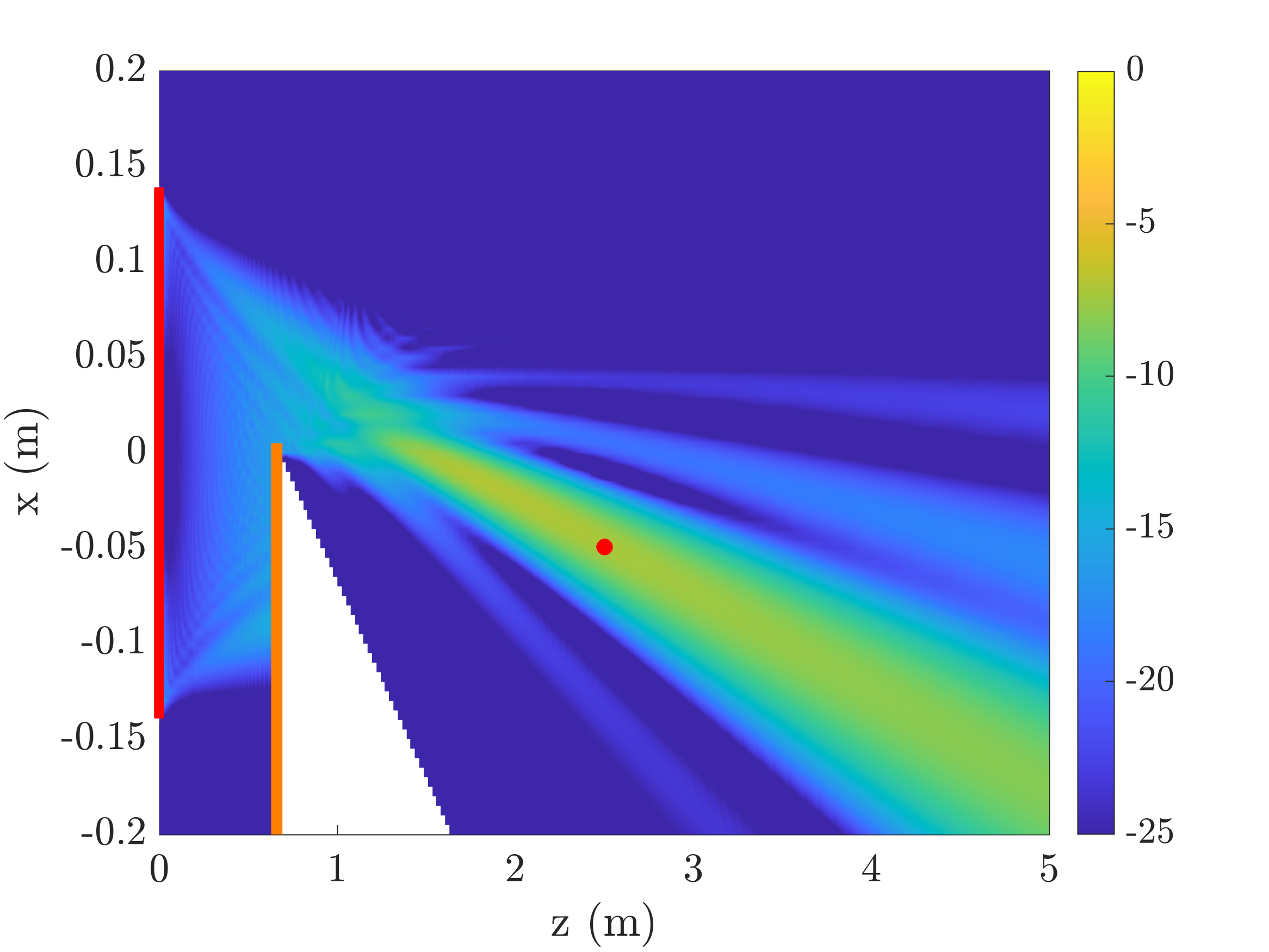}%
\put(61.5,44){\small \color{white} $-7.2$\,dB}
 \put(61,44){\color{white} \vector(-1,-1){8}}
\end{overpic}
\label{fig:blockage4}}
\caption{Example of the normalized power gain (in dB) achieved when transmitting from an $M=256$ antenna ULA at $140$\,GHz toward a marked point with a red circle. Different patterns are observed depending on the precoding (Airy or beamfocusing) and whether there is a blocking object.} \label{fig:blockages} \vspace{-5mm}
\end{figure*}

\vspace{-1mm}
\section{Beamfocusing Versus Airy Beams}
\label{sec:airy}
\vspace{-1mm}

Precoding that generates Airy beams has been described as suitable for quasi-LOS scenarios because it has a mainlobe with a bent locus, in contrast to the ellipsoidal mainlobe of beamfocusing.
In this section, we will compare these beam shapes and begin by describing how Airy beams can be generated by a ULA with analog precoding.
Recall that the $m$:th antenna has coordinate $x_m$ along the $x$-axis. We now let $z>0$ denote the coordinate in the perpendicular direction.

An Airy beam is created using the phase profile \cite[Eq.~(6)]{Zhao2026a}:
\vspace{-0.5mm}
\begin{equation}
 \phi_{\mathrm A}(x)=\frac{(2\pi B)^3}{3}x^3
 -\frac{k}{2F}x^2-k\sin(\vartheta)\,x,
 \label{eq:airy_phase}
\end{equation}
where $B \neq 0$ controls bending, $F>0$ is a focusing-distance parameter, and $\vartheta$ controls steering.
This third-order polynomial is sampled at the antenna coordinates (and the phase is wrapped modulo $2\pi$), leading to the Airy precoder
\vspace{-0.5mm}
\begin{equation}
 [\mathbf w_{\mathrm A}]_m
 =\frac{e^{-\mathrm j\phi_{\mathrm A}(x_m)}}{\sqrt M}.
 \label{eq:airy_precoder}
\end{equation}
Using a Gaussian aperture approximation \cite[Eqs.~(17)--(18)]{Zhao2026a}, the locus of the mainlobe can be expressed in the $xz$-plane as
\vspace{-0.5mm}
\begin{equation}
 x_{\mathrm A}(z)=a_{\mathrm A}-b_{\mathrm A}z-\frac{c_{\mathrm A}}{z},
 \label{eq:Airy_locus}
\end{equation}
using $\xi_{\mathrm p}\approx-1.0188$ (the location of the first maximum of the Airy function) and
\vspace{-0.5mm}
\begin{align}
 b_{\mathrm A}&=\xi_{\mathrm p}\lambda B+\sin(\vartheta)
+\frac{c_{\mathrm A}}{F^2}-c_{\mathrm A}
  \left( \frac{4\lambda}{\pi M^2\Delta^2} \right)^2, \notag\\
 a_{\mathrm A}&=\frac{2c_{\mathrm A}}{F}, \qquad c_{\mathrm A}=\frac{1}{16\lambda\pi^2B^3}.
\end{align}
The locus curve in \eqref{eq:Airy_locus} captures the essence of Airy beams:
\begin{itemize}
\item  \textbf{Nearby beam direction:} When $z$ is small but non-zero, the trajectory approaches the shifted hyperbola $ x_{\mathrm A}^{\textrm{near}}(z)=a_{\mathrm A}-c_{\mathrm A}/z$. The curve should be ignored when it diverges.

\item \textbf{Turning point:} The trajectory changes direction at $z = \sqrt{c_{\mathrm A}/b_{\mathrm A}}$ if $b_{\mathrm A} c_{\mathrm A}>0$, which is the root to $ x_{\mathrm A}'(z)=-b_{\mathrm A}+\frac{c_{\mathrm A}}{z^2}=0$. At the turning point, the sign of $B$ determines the bending direction: $B>0$ implies downward bending, $B<0$ implies upward bending.

\item  \textbf{Faraway beam direction:} When $z$ is large, the trajectory approaches the straight line $ x_{\mathrm A}^{\textrm{far}}(z)=a_{\mathrm A}-b_{\mathrm A}z$. Far-field beam steering matches this direction asymptotically using $w_m = e^{-\mathrm j kx_m\sin(\theta_{\mathrm s})} / \sqrt{M}$, where  $\theta_{\mathrm s}=-\arctan(b_{\mathrm A})$.

\end{itemize}

\subsection{Beam Comparison in a Quasi-LOS MISO Scenario}
\vspace{-1mm}

Returning to the MISO scenario in Section~\ref{sec:blockage}, 
Lemma~\ref{lem:blockage} implies that
 $\gamma(\mathbf w_{\mathrm A};\mathbf D)
 \leq\gamma(\mathbf w_{\mathrm{foc}};\mathbf D)$
when both precoders serve the same receiver with the same power constraint and $\mathbf h$ being known. 
Hence, under the considered blockage model, one cannot enhance the SNR by switching from beamfocusing to an Airy beam, regardless of how the configuration parameters $(B,F,\vartheta)$ are selected. Such a change might, however, modify the beam shape around the receiver, as studied next.

Fig.~\ref{fig:blockages} shows the results of a numerical example adapted from \cite[Fig.~3]{Zhao2026a}, in which a ULA with $M=256$ antennas and $\Delta=\lambda/2$ is used at $140$\,GHz. The aperture length is $27$\,cm. The multiple antennas and beam shaping result in a power gain that we calculate at each location by dividing the actual received power by what it would have been with a single antenna at the origin. The value is shown in dB-scale relative to the theoretical maximum, so that $0$\,dB represents $M=256$.

Fig.~\ref{fig:blockage1} shows the power gain pattern of a transmitted Airy beam with $B=5$\,m$^{-1}$, $F=0.5$\,m, and $\vartheta = -0.03$\,rad. The array is indicated as a red line along the vertical axis, and the approximate trajectory curve from \eqref{eq:Airy_locus} is shown as a dashed curve that roughly follows the actual beam locus.
Next, Fig.~\ref{fig:blockage2} shows the power gain pattern when an object (shown in orange) is placed in front of the lower half of the array at a distance of $0.66$\,m and lets no signal through (i.e., $d_m=0$ for affected antennas).
The mainlobe is barely affected (the Airy configuration was selected to enable this), but a region behind the object cannot be reached at all, and there is also slightly less power gain in some of the upper sidelobes.

Suppose the intended receiver is located at the red circle, which is $z=2.5$\,m from the array and placed on the trajectory curve with $x_{\mathrm A}(2.5)\approx-0.05$\,m.
If our goal is to maximize the SNR at this receiver, we could instead use beamfocusing configured for this location.
The resulting beam pattern is shown in Fig.~\ref{fig:blockage3}.
When the blockage is present, the lower half of the array cannot reach the receiver, so the beam pattern is distorted into the one shown in Fig.~\ref{fig:blockage4}.
Interestingly, in the presence of the blocking object, Airy and beamfocusing achieve visually similar beam shapes around the receiver, while the main differences are elsewhere.
Simply put, the Airy beam uses antennas that cannot reach the receiver to create the curved part of the main lobe in front of the object---the beam does not bend around the object but bends before it.

The power gain at the receiver is $256$ ($0$\,dB) with beamfocusing and $39$ ($-8.1$\,dB) with the Airy beam in the absence of blockage. It drops to $49$ ($-7.2$\,dB) and $38$ ($-8.2$\,dB), respectively, under blockage.
Hence, beamfocusing remains the preferable precoding even when there is a blocking object, as proved by Lemma~\ref{lem:blockage}.
Only $112$ of the $256$ antennas have an unobstructed view of the receiver, so when using beamfocusing, we lose both the blocked antennas' power and contribution to the array gain.
This explains why the power gain becomes $256 \cdot ( \frac{112}{256})^2 = 49$.
By contrast, the power gain of the Airy beam is barely affected by the object, as it spreads its power along a curved trajectory. However, the peak gain is only $15$\% of that with beamfocusing (without blockage).

\subsection{Beam Comparison in a Quasi-LOS MIMO Scenario}
\label{subsec:beamsMIMO}
\vspace{-1mm}

\begin{figure}[t!]
\centering
\includegraphics[trim={0.5cm 0 1.2cm 0.6cm},clip,width=0.9\columnwidth]{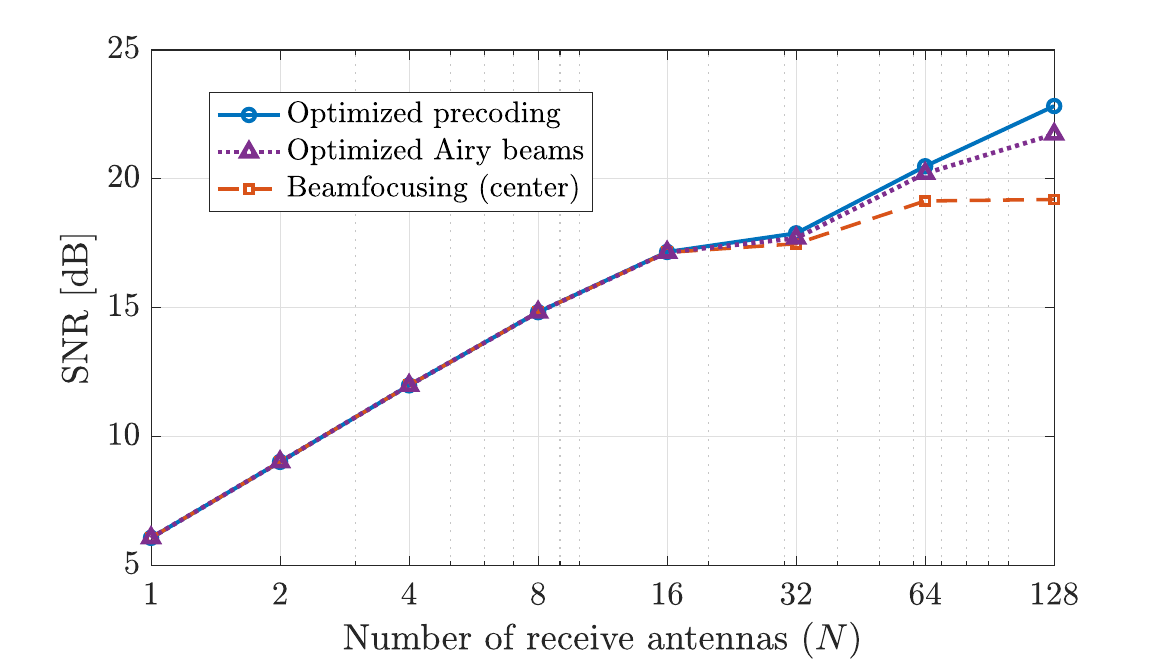} \vspace{-4mm}
\caption{The SNR achieved under blockage for varying number of receive antennas with an $M=256$ antenna ULA transmitter at $140$\,GHz.} \label{fig:MIMO_blockage} \vspace{-6mm}
\end{figure}

We now extend the simulation to consider a receiver with $N\in \{1,2,4,8,16,32,64,128 \}$ antennas.
The antennas are arranged as two parallel half-wavelength ULAs at $140$\,GHz, and there is a blocking object of length $5$\,cm between the arrays, with its center displaced by $-1$\,cm. We compare jointly optimized precoding/combining $\mathbf{w} \in \mathcal{W}, \mathbf{v} \in \mathcal{V}$ (implemented by updating them iteratively until convergence) with numerically optimized Airy precoding parameters $(B,F,\vartheta)$ and fixed beamfocusing from \eqref{eq:focusing} with the receiver's center as the focal point (both with individually optimized combining).

Fig.~\ref{fig:MIMO_blockage} shows the SNR as a function of the number of receive antennas, and without loss of generality, $P/\sigma_n^2$ was selected to achieve $10$\,dB SNR with beamfocusing at $N=1$ without blockage. The SNR generally increases with the number of receive antennas. Further, there is no benefit from deviating from beamfocusing for $N \leq 16$ as the channel matrix is practically rank-one, confirming the analysis in Section~\ref{sec:mimo}.

However, for $N\geq 32$, the receiving ULA is wider than the half-power beamwidth achieved by beamfocusing ($\geq 3.4$\,cm versus $2.1$\,cm), resulting in a multi-rank channel. We then deviate significantly from the conditions in Theorem~\ref{thm:mimo_separable_phases} (including the simplified blockage model) and observe a benefit of optimizing the precoding based on the blocking object to illuminate the unobstructed parts of the receiver. The numerically optimized Airy beam is then better than fixed beamfocusing, as the Airy beam might have any cubic phase profile across the transmitter (it doesn't even have to be curved). Still, we achieve an even higher SNR (up to $\approx 1$\,dB in this example) by optimizing the precoding without the third-order polynomial phase-profile constraint that Airy beams impose.


\section{Conclusion}

In this paper, we examined whether using Airy beams at the transmitter can outperform conventional focused beams in mitigating partial blockage in near-field quasi-LOS communications. 
Under our modeling assumptions, we mathematically and numerically illustrated that \emph{the same beamfocusing-type precoder achieves the highest SNR in quasi-LOS channels} for both single-antenna and small-scale multi-antenna receivers (rank-one channel), as well as for both certain and uncertain channel models. 
The transmitter can even be unaware of the blockage. 
We further noticed that the bending property of the Airy beam family allows us to tailor its trajectory to a specific obstacle, thereby achieving an almost ``blockage-agnostic'' configuration with a similar beam pattern with and without that obstacle. Yet, Airy beams are \emph{not optimal} in terms of the SNR achieved at a receiver behind the obstacle.
The suboptimal performance of bent Airy beams was shown for both small-scale receivers, where they are inferior to beamfocusing, and receivers larger than the beamwidth.
In the latter case, they are outperformed by joint precoder/combiner optimization, which focuses signals on the unobstructed parts of the receiver.

\bibliographystyle{IEEEtran}
\bibliography{IEEEabrv,ref_short_VP}
\vspace{-1mm}

\end{document}